\documentclass[a4paper, USenglish]{lipics-v2018}
 
\usepackage{microtype}
\usepackage[utf8]{inputenc}
\usepackage{cite,nicefrac,amsmath}

\usepackage{amsmath,amsfonts,amssymb,color}
\usepackage{amsmath,amsfonts,amssymb,color}
\usepackage[ruled,vlined,linesnumbered]{algorithm2e}
\usepackage[usenames,dvipsnames]{xcolor}
\usepackage{wrapfig}
\usepackage{graphicx, subcaption}

\newcommand{\rr}{\mathbb{R}}

\newcommand{\nn}{\mathbb{N}}

\renewcommand{\bold}[1]{\textbf{#1}}
\newcommand{\parabold}[1]{\smallskip
\noindent\bold{#1}}

\theoremstyle{plain}
\newtheorem{proposition}[theorem]{Proposition}

\newcommand{\Ex}[1]{\ensuremath{\mathbb{E}\left[#1\right]}}

\newcommand{\D}{\displaystyle}

\newcommand{\MaintainOPT}{\textsc{MaintainOPT}}

\newcommand{\MaintainA}{\textsc{Maintain-}\ensuremath{\mathcal{A}}}

\allowdisplaybreaks

\title{Packing Returning Secretaries}
\author{Martin Hoefer}{Goethe University Frankfurt/Main, Germany}{mhoefer@cs.uni-frankfurt.de}{}{}
\author{Lisa Wilhelmi}{Goethe University Frankfurt/Main, Germany}{wilhelmi@cs.uni-frankfurt.de}{}{}

\authorrunning{M. Hoefer and L. Wilhelmi}

\Copyright{Martin Hoefer and Lisa Wilhelmi}

\subjclass{Theory of computation $\rightarrow$ Design and analysis of algorithms $\rightarrow$ Online algorithms}


\funding{Supported by a grant from the German-Israeli Foundation for Scientific Research and Development (GIF).}

\EventEditors{}
\EventNoEds{1}
\EventLongTitle{}
\EventShortTitle{}
\EventAcronym{}
\EventYear{}
\EventDate{}
\EventLocation{}
\EventLogo{}
\SeriesVolume{}
\ArticleNo{}
\nolinenumbers 

\begin{document}

\maketitle

\begin{abstract}
We study online secretary problems with returns in combinatorial packing domains with $n$ candidates that arrive sequentially over time in random order. The goal is to accept a feasible packing of candidates of maximum total value. In the first variant, each candidate arrives exactly twice. All $2n$ arrivals occur in random order. We propose a simple 0.5-competitive algorithm that can be combined with arbitrary approximation algorithms for the packing domain, even when the total value of candidates is a subadditive function. For bipartite matching, we obtain an algorithm with competitive ratio at least $0.5721 - o(1)$ for growing $n$, and an algorithm with ratio at least $0.5459$ for all $n \ge 1$. We extend all algorithms and ratios to $k \ge 2$ arrivals per candidate.

In the second variant, there is a pool of undecided candidates. In each round, a random candidate from the pool arrives. Upon arrival a candidate can be either decided (accept/reject) or postponed (returned into the pool). We mainly focus on minimizing the expected number of postponements when computing an optimal solution. An expected number of $\Theta(n \log n)$ is always sufficient. For matroids, we show that the expected number can be reduced to $O(r \log (n/r))$, where $r \le n/2$ is the minimum of the ranks of matroid and dual matroid. For bipartite matching, we show a bound of $O(r \log n)$, where $r$ is the size of the optimum matching. For general packing, we show a lower bound of $\Omega(n \log \log n)$, even when the size of the optimum is $r = \Theta(\log n)$.
\end{abstract}

\section{Introduction}

The secretary problem is a classic approach to study optimal stopping problems: A sequence of $n$ candidates are arriving in uniform random order. Each candidate reveals its value only upon arrival and must be decided (accept/reject) before seeing any further candidate(s). Every decision is final -- once a candidate gets accepted, the game is over. Moreover, no rejected candidate can be accepted later on. The goal is to find the best candidate. An optimal solution is to discard the first (roughly) $n/e$ candidates. From the subsequent ones we accept the first that is the best one among the ones seen so far. The probability to hire the best candidate approaches $1/e \approx 0.37$ when $n$ tends to infinity.

The secretary problem and its variants have been popular since the 1960s. Significant interest in computer science emerged about a decade ago due to new applications in e-commerce and online advertising markets~\cite{BabaioffIKK08,HajiaghayiKP04}. For example, the classic secretary problem can be used to model a seller that wants to give away a single item, buyers arrive sequentially over time, and the goal is to assign the item to the buyer with highest value. More generally, online budgeted matching problems arise when search queries arrive over time, and the goal is to show the most profitable ads on the search result pages. The goal here is to design algorithms with good competitive ratio.

More recently, progress has been made towards a general understanding of online packing problems with random-order arrival, including matching~\cite{BabaioffIK07,KorulaP09,KesselheimRTV13}, integer packing programs~\cite{MolinaroR14,KesselheimRTV14}, or independent set problems~\cite{GoebelHKSV14}. Most prominently, the matroid secretary problem has attracted a large amount of interest~\cite{BabaioffIKK08,Dinitz13}. Here the elements of a matroid arrive in uniform random order, and the goal is to construct an independent set with as high a value as possible. A central open problem in the area is the matroid secretary conjecture -- is there a constant-competitive algorithm for every matroid in the random order model? The conjecture has been proved for a variety of subclasses of matroids~\cite{Dinitz13}. Currently, the best-known algorithms for the general problem are $1/O(\log \log \text{rank})$-competitive~\cite{Lachish14,FeldmanSZ18}.

A strong assumption in the secretary problem is that every decision about a candidate must be made immediately without seeing any of the future candidates. Instead, in many natural admission scenarios candidates appear more than once, or they arrive and stay in the system for some time, during which a decision can be made. An interesting variant that captures this idea is the returning secretary problem~\cite{Vardi15}. Here each candidate is assigned two random time points from a bounded time interval. The earlier becomes the arrival time, the later the departure time. Hence, we can assume that each candidate arrives exactly twice, and all $2n$ arrivals occur in random order. The decision about acceptance of a candidate can be made between the first and the second arrival. More generally, for $k \ge 2$ each candidate chooses $k$ random points, arrives at the earliest and leaves at the latest point. In this case, there are $kn$ arrivals in random order. Vardi~\cite{Vardi15} showed an optimal algorithm for the returning secretary problem with $k=2$, for which the probability of accepting the best candidate is about 0.768. For matroid secretary with $k=2$ arrivals, a competitive ratio of 0.5, and for matching secretary a ratio $0.5625 - o(1)$ (with asymptotics in $n$) were shown.

In this paper, we significantly broaden and extend the results on the returning secretary problem towards general packing domains. We provide a simple algorithm that can be combined with arbitrary $\alpha$-approximation algorithms and yields competitive ratios of $0.5 \cdot \alpha$ for all subadditive packing problems, including matroids, matching, knapsack, independent set, etc. Moreover, we improve the guarantees for matching secretary and provide bounds that hold in expectation for all $n$. We extend all our bounds to arbitrary $k \ge 2$. In addition, we study a complementary variant in which the decision maker is allowed to postpone the decision about a candidate. In this case, the goal is to minimize the number of postponements to guarantee an optimal or near-optimal solution in the end. These problems can be cast as a set of novel coupon collector problems, and we provide guarantees and trade-offs for matroid, matching and knapsack postponement.

\subsection*{Results and Contribution}
In the \emph{secretary problem with $k$ arrivals} in Section~\ref{sec:kReturn}, each candidiate arrives exactly $k$ times. We propose a simple approach for general subadditive packing problems with returns, which can be combined with arbitrary offline $\alpha$-approximation algorithms. It yields a competitive ratio of $\alpha \cdot \frac 12$ for $k \ge 2$. For XOS packing problems we show a competitive ratio $\alpha\cdot(1-2^{-(k-1)})$.

For additive bipartite matching, we obtain a new algorithm that provides an improved competitive ratio of $0.5721 - o(1)$ for $k=2$ with asymptotics in $n$. Moreover, we derive an algorithm with ratio $0.5459$ for $k=2$ for every $n$. Both algorithms rely on exact solution of partial matching problems. The algorithms can be combined with faster $\alpha$-approximations for partial matchings, by spending at most an additional factor $\alpha$ in the competitive ratio. For the previous algorithm in~\cite{Vardi15}, the algorithm description and proof of the ratio in the full version is slightly ambiguous.\footnote{For example, the pseudo-code on page 12 does not become the algorithm for a single secretary when there is a single node in the offline partition. One would always accept the best secretary that arrived once in the sample phase. A better one arriving in later rounds is always rejected inside the for-loop. Also, the proof of Claim 5.6 seems to require both sides of the bipartite graph must have size $n$.} Our algorithm clarifies and slightly improves upon this by including the twice-arrived and rejected candidates during a sample phase when computing partial matchings. Their removal yields free nodes in the offline partition for matching in later rounds.

In the postponing secretary problem in Section~\ref{sec:postpone}, there is a pool of $n$ undecided candidates. In each round, a random candidate from the pool arrives. Upon arrival a candidate can be either decided (accept/reject) or postponed and returned into the pool. We strive to minimize the expected number of postponements to compute an optimal or near-optimal solution. Postponing everyone until all candidates are observed at least once is the coupon collector problem. Hence, with an expected number of $O(n \log n)$ postponements we reduce the problem to  thw offline optimization variant. For general XOS packing and an $\alpha$-approximation algorithm, a simple trade-off shows an $(1-\varepsilon)\cdot \alpha$-approximation using $O(n \ln \nicefrac{1}{\varepsilon})$ postponements.

Based on a property we term exclusion-monotonicity, we show significantly improved results when the desired solution has small cardinality. A bound of $O(r \log n)$ for the expected number of postponements holds when obtaining optimal solutions of size at most $r$ in additive matroids and bipartite matching, and greedy 2-approximations for knapsack. For matroids, we can further improve the bound to $O(r' \ln \nicefrac{n}{r'})$, where $r' = \min(r, n-r)$. This upper bound is at most $n$, and the worst-case is attained for uniform matroids. We fully characterize the expected number of postponements of every candidate in uniform matroids when the optimal solution is to be obtained. Finally, we conclude the paper with a lower bound that in general we might need $\Omega(n \log \log n)$ postponements even with an optimal solution of cardinality $O(\log n)$. Due to space constraints, all missing proofs are deferred to Appendix~\ref{app:proofs}.

\subsubsection*{Further Related Work}
The literature on secretary online variants of packing problems and online stochastic optimization has grown significantly over the last decade. We restrict the review to the most directly related results. For a survey of classic variants of the secretary problem, see~\cite{Ferguson89}.

The bipartite secretary matching problem was first studied in the context of transversal matroids~\cite{BabaioffIK07}, where a decision about accepting an arriving vertex into the matching needs to be taken directly, but matching edges can be decided in the end. Later works required that the edges must also be decided upon arrival~\cite{KorulaP09}. The best algorithm for both variants obtains a competitive ratio of $1/e$~\cite{KesselheimRTV13}. Most work in computer science has been devoted to the matroid secretary problem. Currently, the best algorithms obtain a competitive ratio $1/O(\log \log \text{rank})$~\cite{Lachish14, FeldmanSZ18}. It is an open problem if a constant competitive ratio can be shown. For a survey of work on classes of matroids and further developments see~\cite{Dinitz13}. 

While above results are all for maximizing additive objective functions, recent work has started to consider submodular ones. For cardinality and matching constraints, constant-competitive algorithms exist for submodular secretary variants~\cite{KesselheimT17}. For matroids, there is a general technique to extend algorithms for additive objectives to submodular ones, which preserves constant competitive ratios~\cite{FeldmanZ18}.

Beyond matroids and matching, there are constant-competitive algorithms for knapsack secretary~\cite{BabaioffIKK07}. Prominent graph classes in networking applications allow good secretary algorithms for independent set~\cite{GoebelHKSV14}. The techniques for bipartite matching have been extended to secretary variants of combinatorial auctions and integer packing programs~\cite{MolinaroR14,KesselheimRTV14}. Moreover, there are $1/O(\log n)$-competitive algorithms even in a general packing domain~\cite{Rubinstein16}.

Additional model variants that have found interest are, for example, local secretary~\cite{ChenHKLM15} (several decision makers try to simultaneously hire candidates based on local feedback), temp secretary~\cite{FiatGKN15} (candidates are hired only for a bounded period of time), or ordinal secretary~\cite{HoeferK17,SotoTV18} (information available to the decision maker is only the total order of the candidates but not their numerical values).

Secretary postponement can be seen as a combinatorial extension of the coupon collector problem, a classic problem in applied probability. The elementary problem and its analysis are standard and discussed in many textbooks. The problem has many applications, and there is a plethora of variants that have been studied (see, e.g.,~\cite{BonehH97,FlajoletGT92,KobzaJV07}). To the best of our knowledge, however, the results for combinatorial packing problems derived in this paper have not been obtained in the literature before.

\section{Packing Problems}
We consider a \emph{packing problem}, in which there is a set $N$ of $n$ \emph{candidates}, and a set $\mathcal{S} \subseteq 2^N$ of \emph{feasible solutions}. $\mathcal{S}$ is downward-closed, i.e.\ $S \in \mathcal{S}$ and $T \subseteq S$ implies $T \in \mathcal{S}$. For most parts, we assume that the \emph{objective function} $w : 2^N \to \mathbb{R}_{\ge 0}$  is \emph{additive}, i.e.,\ there is a non-negative value $w : N \to \rr_{\ge 0}$ for each candidate, and $w(S) = \sum_{e \in S} w(e)$ for all $S \subseteq N$. More generally, we will sometimes assume the objective function $w$ is in the class \emph{XOS}. An XOS function is defined as $w(S)= \max_{i=1}^k w^i(S)$, the maximum over some number $k$ of additive functions $w^i(S)= \sum_{e \in S} w^i(e)$, for all $S \subseteq N$ and $i=1,\ldots,k$. We also consider functions that are \emph{monotone} ($w(S) \le w(T)$ for $S \subseteq T \subseteq N$) and \emph{subadditive} ($w(S) + w(T) \ge w(S \cup T)$ for all $S,T \subseteq N$). If a packing problem has an $\alpha$-\emph{approximation algorithm}, then for any $N' \subseteq N$ the algorithm guarantees an approximation ratio $\alpha \le 1$ for maximizing $w$ over $\mathcal{S} \cap 2^{N'}$. 

In a secretary variant, we know the number $n$ upfront, and the candidates arrive in random order. Suppose a set $N_i$ of candidates has arrived in rounds $1,\ldots,i$ and candidate $e \in N \setminus N_i$ arrives in round $i+1$. Then $e$ reveals all new feasible solutions with previously arrived candidates $(\mathcal{S} \cap 2^{N_i \cup \{e\}}) \setminus (\mathcal{S} \cap 2^{N_i})$ and their corresponding weight. In the additive case, this simply reduces to revealing the solutions and the weight $w(e)$.

We consider several specific variants. In \emph{matroid secretary}, the set of candidates and the set of feasible solutions form a matroid. Upon arrival, a candidate reveals the new feasible solutions and their weights. In the additive variant with \emph{known matroid}, all candidates and feasible solutions are known upfront. Candidates only reveal their weight upon arrival. 

In \emph{(bipartite) matching secretary}, there is a bipartite undirected graph $(N \cup V, E)$. The nodes in the offline partition $V$ are present upfront. The candidates in the online partition arrive sequentially. The feasible solutions are the matchings in the arrived subgraph. Upon arrival, a candidate reveals its incident edges and weights of the new feasible solutions. In the additive version, the arriving candidate reveals a \emph{weight per edge}, and the weight $w(M)$ of a matching $M$ is the sum of edge weights. Upon accepting a candidate, the algorithm also has to decide which matching edge to include into $M$ (since otherwise it is matroid secretary with transversal matroid).

For (additive) \emph{knapsack secretary}, an arriving candidate $e$ reveals its weight $w(e)$ and a size $b(e) \ge 0$. The size $B$ of the knapsack is known upfront. The feasible solutions are all subsets of candidates such that their total size does not exceed $B$.

\section{Secretaries with $k$ Arrivals}
\label{sec:kReturn}

Suppose that each candidate arrives exactly $k$ times, and all these $kn$ arrivals are presented in uniformly random order. Consider any subadditive secretary packing problem and the following simple algorithm. In the beginning, flip $kn$ fair coins. The number of heads is the length of an initial sample phase. During the sample phase reject all candidates. Consider the set $T$ of candidates that has appeared at least once and at most $k-1$ times in the sample phase. Apply the $\alpha$-approximation algorithm to the instance based on $\mathcal{S} \cap 2^T$ to choose a feasible solution. Accept each candidate in the solution by the time of its $k$-th arrival.

\begin{proposition}
\label{prop:trivial}
	For any XOS packing problem with an $\alpha$-approximation algorithm, the secretary problem with $k$ arrivals allows an algorithm with approximation ratio
	\[
	\beta = \alpha \cdot \left(1 - \frac{1}{2^{k-1}}\right)\enspace.
	\]
	For any monotone subadditive packing problem, the secretary problem with $k \ge 2$ arrivals allows an algorithm with approximation ratio
	\[
	\beta = \alpha \cdot \frac{1}{2}\enspace.
	\]

\end{proposition}

\begin{proof}
Due to random order of arrival, we can simulate generation of $T$ by attaching each of the $kn$ coins to one arrival of one candidate. The arrival is in the sample phase if and only if the coin turns up heads. Then, the probability is $1/2^k$ for each of the following events: (1) a given candidate never appears in the sample phase, and (2) a given candidate appears $k$ times in the sample phase. $T$ is distributed as if we would include each candidate independently with probability $1-\left(\frac{1}{2}\right)^{k-1}$. 

Consider XOS packing problems. Once $T$ is created, we apply the $\alpha$-approximation algorithm to the instance based on $\mathcal{S} \cap 2^T$ to choose a feasible solution. Note that every candidate in $T$ will appear at least once after the sampling phase and therefore is available for acceptance by our algorithm. Each element in $T$ is sampled independently from $N$. Linearity of expectation shows that the optimum $S^*$ restricted to $T$ has expected value $\Ex{w^i(T \cap S^*)} = \left(1-\left(\frac{1}{2}\right)^{k-1}\right) \cdot w^i(S^*)$, for every $i=1,\ldots,k$. Hence, $\Ex{w(T \cap S^*)} \ge \left(1-\left(\frac{1}{2}\right)^{k-1}\right) \cdot w(S^*)$. Now the optimum $S^*_T$ has value $\Ex{w(S_T^*)} \ge \Ex{w(T \cap S^*)}$. By applying the $\alpha$-approximation algorithm to $T$, we obtain a feasible solution $S$ of expected value $\Ex{w(S)} \ge \alpha \cdot \Ex{w(S_T^*)} \ge \alpha \cdot \left(1-\left(\frac{1}{2}\right)^{k-1}\right) \cdot w(S^*)$.

Now consider a monotone subadditive packing problem. Once $T$ is created, for the sake of the analysis we assume a second, hypethetical sampling step: For each candidate $e \in T$ flip an independent coin to remove $e$ from $T$ -- candidate $e \in T$ remains in $T$ with a probability of $\nicefrac{2^{k-2}}{(2^{k-1} - 1)}$. Hence, at the end of the  hypethetical sampling step, each surviving element in $T$ is sampled independently from $N$ with probability 
	\[
	\left(1-\left(\frac{1}{2}\right)^{k-1}\right) \cdot \frac{2^{k-2}}{(2^{k-1}-1)} \quad = \quad \left(1-\left(\frac{1}{2}\right)^{k-1}\right) \cdot \frac{1}{2} \cdot \frac{1}{1-\left(\frac{1}{2}\right)^{k-1}} \quad = \quad \frac{1}{2} \enspace .
	\]
	Due to subadditivity (see~\cite[Proposition 2]{Feige09}), the value of the best feasible solution $S_{T}^* \subseteq T$ has expected value $\Ex{w(S_{T}^*)} \ge \Ex{w(T \cap S^*)} \ge \frac{1}{2} \cdot w(S^*)$. The application of an $\alpha$-approximation algorithm yields a feasible solution $S$ of value $\Ex{w(S)} \geq \alpha \cdot \Ex{w(S_{T}^*)} \geq \frac{\alpha}{2} \cdot w(S^*)$. Obviously the same bound holds when applying the $\alpha$-approximation algorithm to the set $T$ directly without the hypothetical sampling step. 
	\hfill

\end{proof}
The factor $1/2$ in the analysis of our algorithm for monotone sudadditive functions is almost tight. The following example\footnote{We thank an anonymous reviewer for providing the example.} shows a deterioration by almost a factor of 2. Consider a subadditive function $w$ where $w(N)=2$ and $w(S)=1$ for all $S \subset N, N = \{1,\dots,n\}$, where all subsets $S \subseteq N$ are feasible. The optimum is $S^* = N$. However, when rejecting all candidates in the sample phase, there is a high probability to reject at least one candidiate $k$ times -- at which point a reduction by a factor of 2 becomes unavoidable.

Based on this observation, we can very slightly improve the analysis for subadditive functions to obtain a bound that increases monotonically with $k$. With probability $\left(1-\left(\frac{1}{2}\right)^{k-1}\right)^{n}$ the set $T$ contains all elements $N$. In this rare case, the solution of the $\alpha$-approximation algorithm does not suffer from another $1/2$-factor decrease in the approximation ratio. By incorporating this insight into our analysis, the ratio becomes
	\[
	\left(1-\left(\frac{1}{2}\right)^{k-1}\right)^{n} \cdot \alpha + \left( 1 - \left(1-\left(\frac{1}{2}\right)^{k-1}\right)^{n}\right) \cdot \alpha \cdot \frac{1}{2}
	\quad = \quad
	\alpha \cdot \frac{1}{2} \cdot \left( 1 + \left(1-\left(\frac{1}{2}\right)^{k-1}\right)^{n}\right) \enspace,
	\]
	and this is tight for the algorithm in the example above.

For secretary matching, we improve upon this by using a slightly more elaborate approach. The algorithm again samples and rejects a number of candidates that is determined by $kn$ independent coin flips with a suitable probability $p < 1$ (determined below). Hence, the length of the sample phase is distributed according to Binom$(kn,p)$. At the end of the sample phase it computes a matching $M_s$ using an $\alpha$-approximation algorithm for all known candidates and offline vertices $V$. It accepts into $M$ the edges incident to candidates with at most $k-1$ arrivals in the sample. Each of them can be accepted upon their last arrival after the sample phase. The algorithm drops the edges from $M_s$ incident to candidates that arrived $k$ times in the sample. Let $V_s \subseteq V$ be the unmatched offline nodes. 

In the second phase, the algorithm follows ideas from~\cite{KesselheimRTV13,Vardi15}. Upon arrival of a new candidate $e$, the algorithm computes an $\alpha$-approximate matching $M_e$ among $V_s$ and all candidates with first arrival after the sample phase. If $M_e$ contains an edge $(e,v)$ incident to $e$, this edge is added into $M$ if $v$ is still unmatched. Otherwise the edge is discarded.

Since the algorithm can be combined with arbitrary $\alpha$-approximation algorithms for matching, it also applies to, e.g., the $k$-arrival variant of ordinal secretary matching~\cite{HoeferK17}.

\begin{theorem}
For secretary matching with 2 arrivals and any $\alpha$-approximation algorithm for offline matching with $\alpha \leq 1$, there is an algorithm with approximation ratio of $0.5721 \cdot \alpha - o(1)$. For $k$ arrivals, the ratio becomes at least $\alpha \cdot \left(1-\frac{1}{2^{k-1}} + \frac{1}{2^{2k}} - \frac{1}{2^{2k}\cdot(2^k-1)^2}\right) - o(1)$.
\end{theorem}

\begin{proof}
By similar arguments as above, for each arrival of a secretary we can assume to flip a coin independently with probability $p < 1$ that determines if the arrival happens in the sample phase. Hence, each candidate has probability $p^k$ to arrive exactly $k$ times in the sample phase and $(1-p)^k$ to never arrive in the sample phase. Let $M$ be the matching computed by the algorithm, $M_1$ the matching obtained right after the sample phase, $M_2$ the matching composed in the second phase and $M^*$ the optimum matching. It holds $\Ex{w(M)} = \Ex{w(M_1)} + \Ex{w(M_2)}$.

For $M_1$ we interpret the random coin flips as a two-step process. First, for each candidate in $N$ we flip a coin independently with probability $(1-(1-p)^k)$ whether the candidate arrives at least once in the sample phase. Then, we flip another independent coin with probability $p^k/(1-(1-p)^k)$ whether the candidate arrives $k$ times in the sample phase. The first set of coin flips determines the matching $M_s$ that evolves when we apply the $\alpha$-approximation algorithm right after the sample phase. Since every candidate is included independently we have $\Ex{w(M_s)} \ge (1-(1-p)^k) \cdot \alpha \cdot w(M^*)$. Afterwards, the second set of coin flips determines the candidates that are dropped from $M_s$. They are determined independently, so $\Ex{w(M_1)} = \left(1-\frac{p^k}{1-(1-p)^k} \right) \cdot \alpha \cdot \Ex{w(M_s)}$. In total, $\Ex{w(M_1)} \ge (1-(1-p)^k - p^k) \cdot \alpha \cdot w(M^*).$

We denote by $X$ the random number of candidates that arrived at least once during the sample phase. In the second phase phase of the algorithm, we consider all $n-X$ candidates that have not arrived during the sample phase. Standard arguments~\cite{KesselheimRTV13,Vardi15,KesselheimT17} show that each of these newly arriving candidates contributes in expectation a value of $(\alpha \cdot (w(M^*))/n$. For the $\ell$-th first arrival of a new candidate, the probability that the edge $(e,v)$ suggested by the algorithm survives is the probability that the offline node $v \in V$ was not matched earlier, which is lower bounded by
\[
\frac{p^k}{1-(1-p)^k} \cdot \prod_{r=X+1}^{\ell-1} \frac{r-1}{r} 
 \quad = \quad \frac{p^k}{1-(1-p)^k} \cdot \frac{X}{\ell-1}\enspace.
\]
Hence, the expected value for $M_2$ is at least
\begin{align*}
	\Ex{w(M_2) \mid X} 
	&\ge \alpha \cdot w(M^*) \cdot \sum_{\ell=X+1}^{n} \frac{p^k}{1-(1-p)^k} \cdot \frac{X}{\ell-1} \cdot \frac{1}{n} \\
	&\ge \alpha \cdot w(M^*) \cdot \frac{p^k}{1-(1-p)^k} \cdot \frac{X}{n} \cdot \ln \frac{n}{X+1}\enspace.
\end{align*}
For constants $p$ and $k$, standard Hoeffding bounds imply that $X = n(1-(1-p)^k) \pm o(n)$ with probability at least $1-1/n^c$ for suitable constant $c$ (see, e.g.,~\cite{Vardi15}). Hence,
\begin{equation}
  \label{eq:matchBound} 
   \Ex{w(M)}/w(M^*) \quad \ge \quad \alpha \left((1-(1-p)^k - p^k) + p^k \cdot \ln\left(\frac{1}{1-(1-p)^k}\right)\right) - o(1)\enspace,
\end{equation}
where the asymptotics are in $n$. Numerical optimization shows that for $k=2$ and $p \approx 0.49085$, the ratio becomes at least $0.57212 \cdot \alpha - o(1)$. See Table~\ref{tab:ratios} for more numerical results.

\begin{table}[t]
\begin{center}
\renewcommand{\arraystretch}{1.2}
\begin{tabular}{|c||c|c|c|c|c|c|} \hline
$k$ & 2 & 3 & 4 & 5 & 6 & 7 \\ \hline\hline
$p$ & 0.49085 & 0.498901 & 0.499826 & 0.499968 & 0.499994 & 0.499999  \\ \hline
ratio & 0.57212 & 0.766694 & 0.879033 & 0.938491 & 0.968995 & 0.984435 \\\hline 
\end{tabular}
\vspace{0.3cm}
\caption{\label{tab:ratios} Near-optimal parameters $p$ for the sample phase and resulting bounds for the competitive ratio (assuming $\alpha = 1$) derived by numerical optimization of function~\eqref{eq:matchBound}.
}
\end{center}
\end{table}
Intuitively, the algorithm benefits from the unseen candidates after the sample phase and has a tendency to reduce the sample size. On the other hand, the candidates that come $k$ times within the sample phase create the set of free nodes in $V$ available for matching to later candidates. Overall, this leads to a small reduction in the sample size. For larger $k$ this effect vanishes since the number of candidates that appear never or $k$ times during the sample phase both become exponentially small. The optimal sampling parameter quickly approaches $p \to 0.5$. This maximizes the profit from candidates that are available for optimization immediately after the end of the sample phase. Thereby, the improvement over the simple procedure in Proposition~\ref{prop:trivial} becomes smaller.

More formally, we use $\ln (1+x) \ge x-x^2$ in~\eqref{eq:matchBound} and obtain
\[
\Ex{w(M)}/w(M^*) \ge \alpha \left((1-(1-p)^k-p^k) + \frac{p^k\cdot(1-p)^k}{1-(1-p)^k} - \frac{p^k(1-p)^{2k}}{(1-(1-p)^k)^2}\right) - o(1) \enspace.
\]
Note that $\ln(1+x) \le x$, so we deteriorate the expression only by the last negative term. For growing $k$, the optimal value of $p$ approaches $0.5$ very quickly, and we bound
\begin{align*}
\Ex{w(M)}/w(M^*)
&\ge
 \alpha \left(\left(1-\frac{1}{2^k}-\frac{1}{2^k}\right) + \frac{\frac{1}{2^k}\cdot\frac{1}{2^k}}{1-\frac{1}{2^k}} - \frac{\frac{1}{2^k}\cdot\frac{1}{2^{2k}}}{(1 - \frac{1}{2^k})^2}\right) - o(1) \\
&= \alpha \left(1-\frac{1}{2^{k-1}} + \frac{1}{2^{2k}} - \frac{1}{2^{2k}\cdot(2^{2k} - 2^{k+1} + 1)}\right) - o(1)\enspace.
\end{align*}
\vspace{-1.25cm}

\end{proof}

In contrast to~\cite{Vardi15}, our algorithm computes an optimal (or $\alpha$-approximate) matching after the sampling phase for the set of \emph{all} candidates that arrived during that phase (instead of the ones that arrived only once). All candidates that arrived $k$ times are dropped. This creates free nodes of $V$ to be matched to subsequently arriving candidates. The ratios depend asymptotically on $n$, since the guarantee in the second phase relies on concentration bounds for $X$, the number of candidates that arrive at least once in the sampling phase.

Alternatively, one can replace the second phase by recursively applying the sampling phase. More formally, after the sampling phase is done and matching $M_1$ is added to $M$, we apply the same sampling phase to $V_s$ and the candidates that have not arrived so far. In this way, we can iterate the sampling step and re-apply it to the unseen candidates and left-over nodes of the offline partition. The resulting ratios do not require concentration bounds.

\begin{corollary}
\label{cor:kMatching}
For secretary matching with 2 arrivals and any $\alpha$-approximation algorithm for offline matching with $\alpha \leq 1$, there is an algorithm with approximation ratio of $0.5459\cdot\alpha$ for every $n \ge 1$. For $k$ arrivals, the ratio becomes at least $(1 - \frac{1}{2^{k-1}} + \frac{1}{2^{2k}} - \frac{2^k-1}{2^{2k} \cdot(2^{2k} - 2^{k} - 1)}) \cdot \alpha$ for every $n \ge 1$.
\end{corollary}


\section{Postponing Secretaries}
\label{sec:postpone}

Now suppose that for each arriving candidate the algorithm can decide (accept/reject) or postpone it. The goal is to compute an optimal or near-optimal solution with a small expected number of postponements. Consider any algorithm for the postponement problem. We cluster the execution into rounds. Round $i$ are the arrivals from and including the $i$-th \emph{unique arrival} (i.e., the $i$-th time a candidate arrives for the first time) and before the $(i+1)$-th unique arrival. Clearly, there are always $n-1$ rounds in the execution of any algorithm. If we simply postpone every candidate until we have seen all $n$ candidates, we have full information to make accept/reject decisions for all candidates. Then the problem reduces to the classic coupon collector problem, and the expected number of returns is $\Theta(n \log n)$. Our goal is to examine how we can improve upon this baseline.

We first consider a general result for XOS packing. To reduce the expected number of returns to $\Theta(n)$, it is sufficient to sacrifice a constant factor in the approximation ratio. We obtain the following FPTAS-style trade-off between postponements and solution quality. 
\begin{proposition}
\label{prop:trivial2}
For any\footnote{For $\varepsilon \le 2/n$, the bound on the expected number of postponements remains $\Theta(n \log n)$ by simply observing all applicants and computing an $\alpha$-approximation.} $\varepsilon > 2/n$ and any XOS packing problem with $\alpha$-approximation algorithm, there is an $\alpha \cdot (1-\varepsilon)$-approximation algorithm with an expected number of postponements of $\Ex{R} < n \cdot \ln (2/\varepsilon)$.
\end{proposition}

\begin{proof}
We postpone every candidate until round $\lceil n(1-\varepsilon) \rceil$. Then, we run the $\alpha$-approximation algorithm on the subset of arrived candidates. By the same arguments as in Proposition~\ref{prop:trivial}, this yields an $\alpha(1-\varepsilon)$-approximation.

Let $R^i$ be the number of postponements in round $i$. Clearly, by linearity of expectation, $\Ex{R} = \sum_{i=1}^{n-1} \Ex{R^i}$. In each round, the number of postponements is the number of rounds until we see the next unique arrival, and, hence, distributed according to a negative binomial distribution. Therefore, their expected number is
\begin{align*}
\Ex{R} &\le \sum_{i=1}^{\lceil n(1-\varepsilon) \rceil} \left(\frac{n}{n-i} - 1 \right) = n \cdot \sum_{i=1}^{\lceil n(1-\varepsilon)\rceil} \frac{1}{n-i} - \lceil n(1-\varepsilon) \rceil \\
&\le n \cdot (\ln (n) - \ln(n\varepsilon-1) - 1 + \varepsilon) \le n \cdot (-\ln(\varepsilon - 1/n)) < n \ln\left( \frac{2}{\varepsilon}\right)\enspace.
\end{align*} 
\vspace{-1.25cm}

\end{proof}
A similar (but not exactly FPTAS-style) analysis applies to monotone subadditive packing problems. We can sample candidates until round $\lceil n/k \rceil$, for an integer $k = 2,3,4,\ldots$. Applying the $\alpha$-approximation algorithm to the subset of arrived candidates we compute the approximate solution $S$. Based on~\cite[Proposition 2]{Feige09}, this solution represents an $\alpha/k$-approximation. The expected number of postponements is
	\begin{align*}
		\Ex{R} &\le \sum_{i=1}^{\lceil n/k \rceil} \left(\frac{n}{n-i} - 1 \right) = n \cdot \sum_{i=1}^{\lceil n/k \rceil} \frac{1}{n-i} - \lceil n/k \rceil \le n \cdot \ln\left(\frac{k}{k-1}\right)\enspace.
\end{align*}
\subsection{Exclusion-Monotone Algorithms}

We obtain significantly better guarantees for packing problems and algorithms with a monotonicity property. Consider a packing problem and any algorithm $\mathcal{A}$. We denote by $\mathcal{A}(T)$ the solution computed by $\mathcal{A}$ when applied to $T \subseteq N$. 

\begin{definition}
 A sequence of subsets $(N_i)_{i \in \nn}$ with $N_i \subseteq N$ is called inclusion-monotone if $N_i \subseteq N_j$ for all $i \le j$. An algorithm $\mathcal{A}$ is called \emph{$r$-exclusion-monotone} if for every inclusion monotone sequence there is a sequence of subsets $(D_i)_{i \in \nn}$ with $\mathcal{A}(N_i) \subseteq D_i \subseteq N_i$, $|D_i| \le r$ and $N_i \setminus D_i \subseteq N_j \setminus D_j$ for all $i \le j$.
\end{definition}

Intuitively, to determine its solution for any subset of available elements $N_i$, an $r$-exclusion-monotone algorithm $\mathcal{A}$ can restrict attention to a set $D_i$ of at most $r$ elements. Moreover, $\mathcal{A}$ is such that any element $e \in N_i \setminus D_i$ that is discarded must never be reconsidered when more elements become available.

This property is exhibited in a variety of important packing domains. For these problems we can obtain more fine-grained, significantly improved guarantees based on solution size.

\begin{proposition}
\label{prop:exclusionMonotone}
The following algorithms are $r$-exclusion-monotone.
\begin{itemize}
\item Optimal algorithm \textsc{Opt} for matroids. $r$ is the rank of the matroid.
\item Optimal algorithm \textsc{Opt} for bipartite matching. $r$ is the maximum cardinality of any matching\footnote{Recall that \emph{vertices in one partition} arrive and get postponed, along with their incident edges. If \emph{single edges} arrive and must be postponed individually, the property might not hold (c.f.\ Example~\ref{ex:knownGraph} below).}.
\item \textsc{Greedy} 0.5-approximation algorithm for knapsack. Here $r = |S| + 1$ with $S$ a feasible packing of the knapsack with maximum cardinality.
\end{itemize}
\end{proposition}

Now consider candidates arriving in random order with postponements. Obviously, the set of arrived candidates forms an inclusion-monotone sequence. In our algorithm \MaintainA, we apply the $r$-exclusion-monotone algorithm $\mathcal{A}$ in the beginning of round $i$ to the set $N_i$ of arrived candidates. \MaintainA\ immediately rejects any candidate as soon as it is not contained in $D_i$. It keeps postponing the candidates in $D_i$. Finally, \MaintainA\ accepts the candidates in $\mathcal{A}(N)$ after the last round. Note that for the following result, \MaintainA\ does not have to know $n$, $r$ or any properties of the unseen candidates. The following guarantee significantly improves over the simple bound given in Proposition~\ref{prop:trivial2} when the solution is drawn from a small subset of elements.

\begin{theorem}
  \label{thm:exclusionMonotone}
  Consider a packing problem with an $r$-exclusion-monotone $\alpha$-approximation algorithm $\mathcal{A}$. The corresponding algorithm \MaintainA\ computes an $\alpha$-approximation with an expected number of postponements $\Ex{R} = \Theta(r \ln n/r')$, where $r' = \min(r,n-r)$. 
\end{theorem}

\begin{proof}
Consider the execution of the algorithm in rounds as discussed above. In each round, let $U_i$ denote the number of candidates that are still undecided (i.e., either have not arrived or have been left undecided in earlier rounds). In round $i$ we have seen exactly $i$ candidates. Thus, given $U_i$ undecided candidates, the expected number of postponements $R^i$ in round $i$ is given by a negative binomial distribution and amounts to
\[
\Ex{R^i \mid U_i} = \left(\frac{U_i}{n-i} - 1\right)\enspace.
\]
To bound $U_i$ we note that, trivially, $U_i \le n$. Moreover, the number of candidates that have arrived and are undecided is $U_i - (n-i)$. Since \MaintainA\ postpones only candidates in the set $D_i$, we have that $U_i - (n-i) \le r$. This implies $U_i \le \min(n, n-i+r)$ and yields
\begin{align*}
\Ex{R}  &\le \sum_{i=1}^{r-1} \left( \frac{n}{n-i} - 1\right) +  \sum_{i=r}^{n-1}\left(\frac{n-i+r}{n-i} - 1\right) \\
&= n\sum_{i=1}^{r-1} \frac{1}{n-i} -r +  r\sum_{i=r}^{n-1}\frac{1}{n-i} \\
&\le n \left(\frac{1}{n-r+1} + \ln\left(\frac{n-1}{n-r+1}\right)\right) + r \left(\frac{1}{r} - 1 + \ln\left(\frac{n-1}{r}\right) \right)\\
&= \left(2 + \frac{r-1}{n-r+1} - r\right) + n\ln\left(\frac{n-1}{n-r+1}\right) + r \ln\left(\frac{n-1}{r}\right)\enspace.
\end{align*}

Clearly, the first term in the bracket is at most 1. For $r \ge n-r+1$, the second term is larger than the third term and amounts to $O(r \ln \nicefrac{n}{r'})$. For $r \le n-r+1$, we upper bound 
\[ n\ln\left(\frac{n-1}{n-r+1}\right) = n\ln\left(1 + \frac{r-2}{n-r+1}\right)\le (r-2) + \frac{(r-2)(r-1)}{n-r+1} < 2r-4\enspace. \] 
Thus, the asymptotics are dominated by the third term, and $\Ex{R} = O(r \ln \nicefrac{n}{r'})$. A similar calculation using elementary lower bounds shows that $\Ex{R} = \Omega(r \ln \nicefrac{n}{r'})$.
\end{proof}

\subsection{Matroids}
We adjust \MaintainA\ for \emph{known matroids}, i.e.\ when the structure of the matroid is known upfront (only the weights of the elements are revealed). In this case, we can assume $r \le n/2$, since for $r \ge n/2$ we can consider finding a minimum-weight basis in the dual matroid. We adjust algorithm \MaintainA\ in the following way and denote the resulting algorithm by \MaintainOPT. Instead of postponing all elements in the current optimum until the end, we can accept some elements earlier. In particular, we can directly accept an element $e$ as soon as there is no unseen element that can force $e$ to leave the optimum solution. This allows to significantly improve the number of returns to below $n$ for any rank of the matroid.

\begin{theorem}
  \label{thm:matroids}
  For the class of all matroids with rank $r$, the expected number of postponements $R$ in \MaintainOPT\ with known matroid is maximized for the uniform matroid. It is bounded by $\Ex{R} = \Theta(r' \ln \nicefrac{n}{r'})$, where $r' = \min(r, n-r)$. For every matroid it holds that $\Ex{R} < n$.
\end{theorem}

Note that for any postponement problem, a simple calculation shows that the expected number of postponements of any single candidate can always be upper bounded by $\ln n + 1$. 
%
In contrast, the previous theorem shows that, on average, we need less than one postponement per candidate to compute even an optimal solution in matroids. However, they can be quite unbalanced over the candidates. We fully characterize the expected number of postponements in the uniform matroid with $r\le n/2$. The worst candidate in the optimal solution (i.e., the $r$-th best candidate) asymptotically gets the largest expected number of postponements. The expected number is decreasing quickly for better and worse candidates.

\begin{theorem}
\label{thm:matroidsAll}
For \MaintainOPT\ with known uniform matroid of rank $r \le n/2$, the expected number of postponements $R_j$ of the $j$-th best candidate is bounded by
\[ \Ex{R_j} \le 
    \begin{cases}
    \renewcommand{\arraystretch}{1.5} 
    \D \ln\left(\frac{n-j}{r-j+1}\right) + O(1) \enspace, &\hspace{0.25cm} \text{ for } j = 1, \ldots, r,\\[0.5cm]
    \D \ln\left(\frac{j-1}{j-r}\right) + O(1) \enspace, &\hspace{0.25cm} \text{ for } j = r+1, \ldots, n.
    \end{cases}
\]
\end{theorem}

\begin{figure}[t]
	\begin{subfigure}[t]{0.32 \linewidth}
		\includegraphics[width=\linewidth, height = 0.14 \textheight]{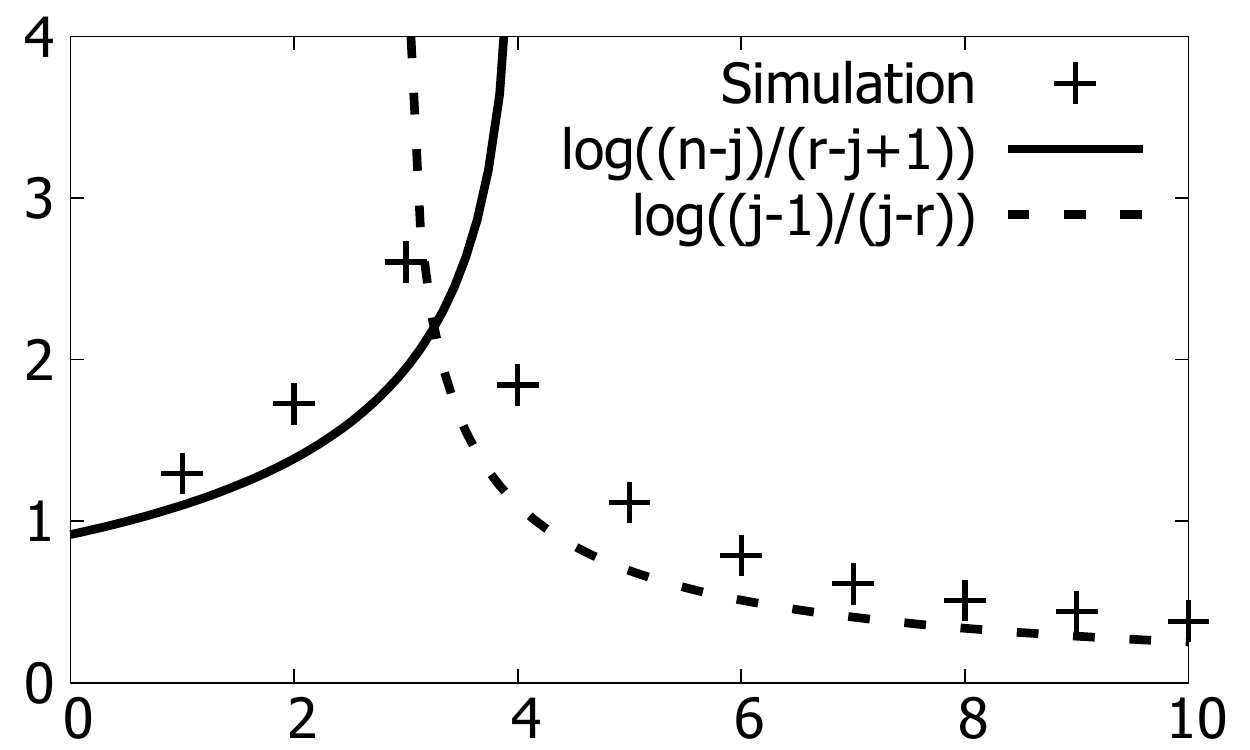}
	\end{subfigure}
	\begin{subfigure}[t]{0.32 \linewidth}
		\includegraphics[width=\linewidth, height = 0.14 \textheight]{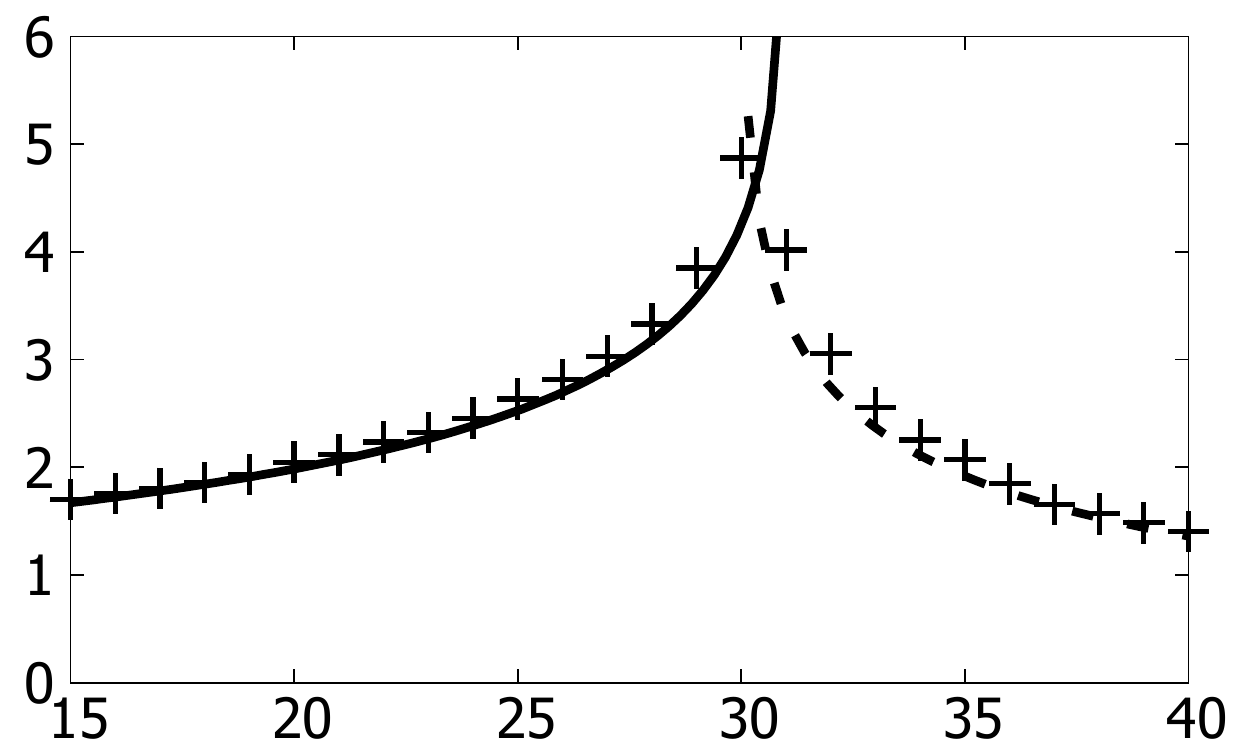}	
	\end{subfigure}
	\begin{subfigure}[t]{0.32 \linewidth}
		\includegraphics[width=\linewidth, height = 0.14 \textheight]{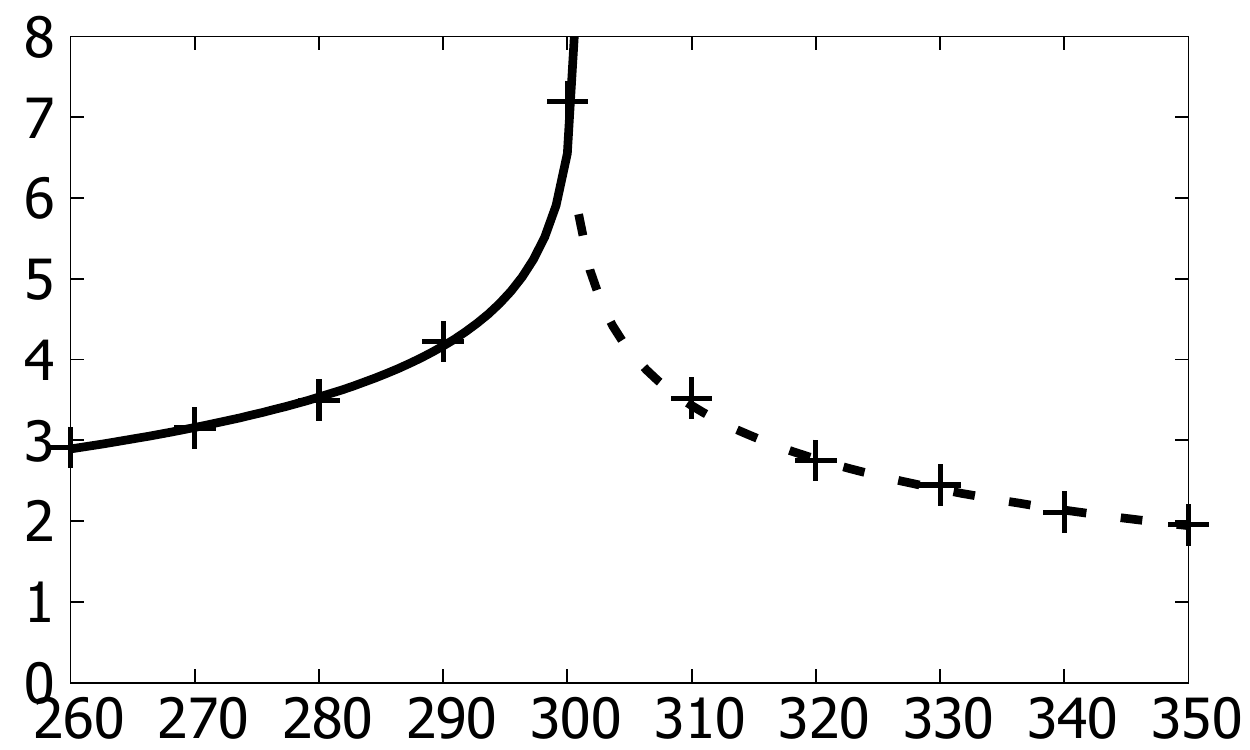}
	\end{subfigure}
	\caption{\label{fig:postponement} Number of postponements of \MaintainOPT\ in a uniform matroid with $n=10$ and $r = 3$ (left), $n=100$ and $r=30$ (middle), and $n=1000$ and $r = 300$ (right). The x-axis is the index of the candidate in the sorted order. The y-axis shows the average number of postponements over 5000 runs. The $O(1)$-terms in Theorem~\ref{thm:matroidsAll} turn out to be small. They appear to be maximal for candidates $r$ and $r+1$, but seem to vanish for growing $n$.}
\end{figure}
Based on our experiments in Figure~\ref{fig:postponement} the $O(1)$ terms are small and even seem to vanish for large $n$. The logarithmic function captures the number of postponements rather precisely.

For matroids, the number of postponements of \MaintainOPT\ with known matroid is always at most $n$. Instead, for bipartite matching the number of postponements of \MaintainOPT\ must grow to $\Theta(n \log n)$ when $r$ becomes large, even if the graph is known.

\begin{example}
\label{ex:knownGraph}
Consider a simple cycle of length $2n$ and number the vertices consecutively around the cycle. Suppose the $r = n$ even vertices form the offline partition $V$, and the $n$ odd vertices arrive in random order. The edge weights can be arbitrary, but an adversary chooses them to be in $[1, 1+\varepsilon]$. Then, unless we see all vertices, we cannot decide which of the two perfect matchings will be the optimal one. \MaintainOPT\ needs to see all vertices to be able to decide the matching edges. We recover the coupon collector problem. \hfill $\blacksquare$
\end{example}

The example also applies when the \emph{edges} of the bipartite graph are candidates that arrive in random order (rather than the vertices). In order to guarantee that an optimal solution is returned with probability 1 in the end, all $2n$ candidate edges need to remain undecided until the last unique arrival. This shows, in particular, that the bound of $O(r' \ln \nicefrac{n}{r'})$ for \MaintainOPT\ for known matroids cannot be extended to known \emph{intersections of matroids}.

\subsection{Exclusion-Monotonicity and Solution Size} 
For $r$-exclusion-monotone algorithms $\mathcal{A}$ the algorithm \MaintainA\ needs at most $O(r \ln n)$ postponements. One might hope that for any $r$-exclusion-monotone algorithm the parameter $r$ is tied closely to the solution size of the algorithm. Then a large number of returns in \MaintainA\ would be caused by $\mathcal{A}$ returning a solution with many elements. This, however, is not the case -- even if we are guaranteed that the size of the optimal solution is $\Theta(\log n)$, an expected number of $\Omega(n \log \log n)$ postponements for \MaintainOPT\ can be required.

\begin{theorem}
\label{thm:logToNloglog}
There is a class of instances of the independent set problem with every optimal solution of size $|I^*| = 3 \ln n$, for which the expected number of postponements $R$ in \MaintainOPT\ is $\Ex{R} = \Omega(n \ln \ln n)$.
\end{theorem}







%



\bibliographystyle{abbrvurl}


\clearpage

\appendix

\section{Omitted Proofs}
\label{app:proofs}

\subsection{Proof of Corollary \ref{cor:kMatching}}

We apply the sampling phase recursively on the unknown candidates that have not arrived and the nodes $V_s$ that are still unmatched. In this way, we obtain more phases, and we denote the matching edges added in phase $i$ by $M_i$. In total, the matching $M$ computed by the algorithm is composed of $M = M_1 \cup M_2 \cup M_3 \cup \ldots$, and $\Ex{w(M)} = \sum_{i=1}^\infty \Ex{w(M_i)}$. For the first phase, we already argued above that $\Ex{w(M_1)} \ge (1-(1-p)^k -p^k) \cdot \alpha \cdot w(M^*)$. 

For the second phase, we consider the set of candidates that have not arrived during the first phase. For each such candidate and each of its $k$ arrivals, we can again assume to throw another random coin with probability $p$ if at least one arrival is in the second sample phase. Thus, every candidate arrives at least once in the second sample phase with probability $(1-p)^k\cdot(1-(1-p)^k)$. For the offline partition, we can assume that each node survives the first phase independently with a probability of at least $(p^k/(1-(1-p)^k))$.  This value is exactly the probability of matching the node to a candidate with $k$ arrivals in the first phase. Thus, the $\alpha$-approximate matching $M_{s,2}$ based on the candidates that arrive for the first time in the second sample phase has value at least
\begin{align*} 
\Ex{w(M_{s,2})} &\ge \alpha \cdot w(M^*) \cdot (1-p)^k\cdot(1-(1-p)^k)\cdot \frac{p^k}{1-(1-p)^k}\\
&= (1-p)^k p^k \cdot \alpha \cdot w(M^*)\enspace. 
\end{align*}
Now, each candidate that arrived at least once in the second sample phase is dropped when he arrived $k$ times during the second sample phase. This happens independently with probability $p^k/(1-(1-p)^k)$. Thus, for the expected value of $M_2$ it holds
\begin{align*} 
\Ex{w(M_2)} &\ge \left(1 - \frac{p^k}{1-(1-p)^k}\right) \cdot \Ex{w(M_{s,2})} \\
&\ge (1-(1-p)^k-p^k) \cdot \frac{(1-p)^kp^k}{1-(1-p)^k} \cdot \alpha \cdot w(M^*)\enspace. 
\end{align*}
Iterating this argument for the subsequent recursions, we see that
\begin{align}
\Ex{w(M)} &= \sum_{i=1}^\infty \Ex{w(M_i)} \nonumber \\
&\ge \alpha \cdot w(M^*) \cdot (1-(1-p)^k-p^k) \cdot \sum_{i=0}^\infty \left( \frac{p^k (1-p)^k}{1-(1-p)^k} \right)^i \nonumber \\
&= \alpha \cdot w(M^*) \cdot (1-(1-p)^k-p^k) \cdot \frac{1-(1-p)^k}{1-(1-p)^k - p^k (1-p)^k } \label{eq:matchBound2}\enspace.
\end{align}
Numerical optimization shows that for $k = 2$ and $p = 0.48638$ we obtain a ratio of at least $0.5459$. More numerical results are shown in Table~\ref{tab:ratios2}.

\begin{table}[t]
\begin{center}
\renewcommand{\arraystretch}{1.2}
\begin{tabular}{|c||c|c|c|c|c|c|} \hline
$k$ & 2 & 3 & 4 & 5 & 6 & 7 \\ \hline\hline
$p$ & 0.48638 & 0.498133 & 0.49968 & 0.499939 & 0.499988 & 0.499997  \\ \hline
ratio & 0.5459 & 0.763646 & 0.87866 & 0.938445 & 0.96898 & 0.984435 \\\hline
\end{tabular}
\vspace{0.3cm}
\caption{\label{tab:ratios2} Near-optimal parameters $p$ for the sampling phases and resulting bounds for the competitive ratio derived by numerical optimization of function~\eqref{eq:matchBound2}.}
\end{center}
\end{table}

For large $k$, the optimal value for $p$ rapidly approaches $1/2$. For $p = 1/2$ we obtain
\begin{align*}
\Ex{w(M)}/w(M^*) &\ge \alpha \cdot \left(1 - \frac{1}{2^{k-1}}\right) \cdot \frac{1-\frac{1}{2^k}}{1-\frac{1}{2^{k}} - \frac{1}{2^{2k}}} \\
 &= \alpha \cdot \left(1 - \frac{1}{2^{k-1}} + \frac{1}{2^{2k}} - \frac{2^k-1}{2^{2k} \cdot(2^{2k} - 2^{k} - 1)}\right)\enspace.
\end{align*}
Note that these ratios do not require concentration bounds. They apply for the expected value of the matching for any number $n$ of candidates. There are no low-order terms $o(1)$.
\qed

\subsection{Proof of Proposition~\ref{prop:exclusionMonotone}}
For matroids, upon arrival of an additional element, the new element forms a fundamental circuit with respect to the current optimal basis. A new optimal basis can be computed by discarding the element of smallest weight from the circuit. Since the new optimum can be computed from the old optimum and the newly arrived element, we never have to return discarded elements into the optimal solution. Hence, we can use $D_i$ as the set of elements in the optimum.

For matching, upon arrival of an additional node $v$, consider the symmetric difference between the old optimal matching $M$ and the new optimal matching $M'$. There is exactly one augementing path starting with the newly arrived node and ending with a node from the offline partition or a node from the online partition that is in the current optimal matching. Hence, we can use $D_i$ as the set of online nodes in the optimal matching.

For knapsack, upon arrival of an additional element $e$, \textsc{Greedy} composes a feasible subset $S$ of elements by greedily packing elements in non-increasing order of ratio weight/size. In the end, the solution is either $S$ or the element of maximum weight $e_{\max}$, whichever gives more value. Hence, $D_i$ can be restricted to the new element $e$, the arrived element of maximum-weight $e_{\max}$, and the greedy set $S$.
\qed

\subsection{Proof of Theorem~\ref{thm:matroids}}
Again we analyze the process in rounds as above. Given $U_i$ undecided candidates, the expected number of postponements $R^i$ in round $i$ is given by a negative binomial distribution and amounts to $\Ex{R^i \mid U_i} = \left(\frac{U_i}{n-i} - 1\right)$. We obtain the bounds $U_i \le n$ and $U_i - (n-i) \le r$ as in the proof of Theorem~\ref{thm:exclusionMonotone}. Due to the matroid property, if there are at most $x$ unseen candidates, they can cause at most a set of $x$ candidates to leave the optimum solution. This set can be determined in polynomial time (see Proposition~\ref{prop:unseen} below for a proof of this fact). Hence, \MaintainOPT\ must keep always at most $(n-i)$ arrived candidates undecided, i.e., $U_i - (n-i) \le (n-i)$. Overall, this implies
\[
U_i \le \min\{n, \; n-i+r, \; 2(n-i)\}\enspace.
\]
This upper bound can be used to divide the process into three phases. Phase 1 consists of the rounds $i=1,\ldots,r-1$ where $U_i \le n$. Phase 2 consists of the rounds $i=r,\ldots,n-r$, where $U_i \le n-i+r$. Finally, Phase 3 are the rounds $i = n-r+1, \ldots,n-1$, in which $U_i \le 2(n-i)$.

Observe that these upper bounds on $U_i$ are always attained exactly when we apply \MaintainOPT\ in the uniform matroid. In Phase 1, no candidate can be accepted or rejected and all arriving candidates get postponed. The number of undecided candidates stays $U_i = n$. In Phase 2, only the $r$ candidates in the current optimum solution are both arrived and undecided, so $U_i = n-i+r$. In Phase 3, upon each arrival of a previously unseen candidate, the algorithm makes accept/reject decisions. It can accept exactly one candidate and reject exactly one other candidate in each of the rounds of Phase 3, and thus $U_i = 2(n-i)$. 

This proves that the uniform matroid results in the largest expected number of postponements, which is given by
\begin{align*}
\Ex{R} &= \sum_{i=1}^{r-1} \left(\frac{n}{n-i} - 1\right) + \sum_{i=r}^{n-r}\frac{r}{n-i} + \sum_{i=n-r+1}^{n-1} 1\\
&= n \sum_{i=1}^{r-1} \frac{1}{n-i} + r \sum_{i=r}^{n-r} \frac{1}{n-i}\\
&\le n \left(\frac{1}{n-r+1} + \ln\left(\frac{n-1}{n-r+1}\right)\right) + r \left(\frac{1}{r} + \ln\left(\frac{n-r}{r} \right)\right)\\
&\le \left(2 + \frac{r-1}{n-r+1}\right) + \frac{n(r-2)}{n-r+1} + r \ln\left(\frac{n}{r} - 1 \right)\\
&= \frac{(r-1)^2}{n-r+1} + r \left(1 + \ln\left(\frac{n}{r} - 1 \right)\right)\enspace.
\end{align*}
For $r \le n/2$, this clearly shows $\Ex{R} = O(r \ln \nicefrac{n}{r})$.  Note that the expression is monotone in $r$, and for $r=n/2$ we obtain $\Ex{R} \le \frac{(n/2-1)^2}{n/2 + 1} + \frac{n}{2} < n$. By lower bounding $\sum_{i=x}^y \frac{1}{n-i} \ge \ln(\frac{n-y}{n-x+1})$ and $\ln(1+x) \ge x/\ln(2)$ for $x \in [0,1]$, a very similar calculation shows $\Ex{R} = \Omega(r \ln \nicefrac{n}{r})$. 
\qed

\subsubsection{Weight Revelation and Basis Adjustment}
For completeness, we analyze the size and structure of the set of elements that can be forced to leave an optimum solution upon addition of other elements. Consider any weighted matroid of rank $r$. Let $K$ denote the set of elements with known weight, and let $U$ be the set of elements with unknown weight. W.l.o.g.\ we assume all weights are pairwise distinct and strictly positive. Let $I^*_K \subseteq K$ be an independent set of maximum weight. For any vector of weights $(w(u))_{u \in U}$, consider the set of elements $I_w \subseteq I^*_K$ that are not part of the optimal basis after the weights $w(u)$ are revealed. We show that $\mathcal{K} = \bigcup_{w \in \mathbb{R}^{|U|}} I_w$ has cardinality $|\mathcal{K}| = \min(|U|,|I_K^*|)$.

\begin{proposition}
  \label{prop:unseen}
  The set $\mathcal{K}$ of elements that can be forced to leave the optimal independent set $I_K^*$ upon revelation of the unknown weights for elements in $U$ has cardinality at most $\min(|U|,|I_K^*|)$.
\end{proposition}

\begin{proof}
  If $I_K^*$ is not a basis of the matroid, then for convenience we add enough dummy elements (e.g., copies of the elements in $U$) with tiny weight so that $I_K^*$ is a basis. Suppose we now reveal the weights of elements in $U$ one by one. Upon revelation of its weight, an element $e \in U$ enters the optimal basis if and only if $w(e) > \min_{e' \in C} w(e')$, where $C$ is the unique fundamental circuit formed by $e$ and the current optimal basis. It forces the unique weight-minimal element in $C$ to leave the basis. Hence, when considering the sets $I_w$ that can be forced to leave the optimal basis, we can assume that all elements $e \in U$ have either weight $w(e) = 0$ or $w(e) > \max_{e' \in K} w(e')$.

  Consider some maximal set of elements $U' \subseteq U$ that can enter the optimal basis. The above observation implies that $U'$ forces a unique set $K' \subseteq K$ with $|K'| = \min(|U'|,I^*_K)$ many elements to leave the basis. We will show that $\mathcal{K} = K'$, i.e., this is the unique (super-)set of elements that can be forced to leave the basis by \emph{any} subset of elements from $U$.

  Suppose only $U'' \subset U'$ enters the optimal basis. We reveal the weights of $U''$ first, and assume they all enter the basis. Suppose for contradiction that a set $K'' \not\subset K'$ leaves the optimal basis. We could go on and reveal the weights of $U'' \setminus U'$, and they could subsequently enter the basis as well -- but by the updates via fundamental circuits we would never recover any element from $K'' \setminus K'$. This would contradict the fact that $U'$ forces the set $K'$ to leave. Hence, addition of $U''$ forces a set $K'' \subset K'$ to leave $I_K^*$.

  Consider a different maximal set $U'' \subseteq U$ that can enter the optimal basis. We construct the optimal basis containing $U''$ as follows. We first reveal the weights of $U'$, and we assume they all enter the optimal basis. Every element $u \in U \setminus U'$ forms a unique fundamental circuit with the optimal basis $(I_K^* \cup U') \setminus K'$. Since $U'$ is maximal, none of these circuits contains an element from $K$. Now we reveal the weights from $U'' \setminus U'$ one by one. We assume they are large enough, so all elements enter the optimal basis. Since initially none of the circuits contains an element from $K$, the strong circuit elimination axiom shows that over the course of the revelation process, no circuit with an element from $K$ can evolve. The elements in $U'' \setminus U'$ only force elements from $U' \setminus U''$ to leave the optimal basis. Hence, addition of $U''$ forces the set $K'$ to leave $I_K^*$.
\end{proof}

\subsection{Proof of Theorem~\ref{thm:matroidsAll}}
We again consider the process in rounds. Let us first concentrate on the top candidates.

\parabold{Top-$r$ candidates:}
For the top $r$ candidates, we again consider three phases (phase 1: every candidate is postponed, rounds $1,\ldots,r-1$; phase 2: exactly the top $r$ candidates that have arrived so far are postponed, rounds $r,\ldots, n-r$; phase 3: one candidate is accepted and one is rejected, rounds $n-r+1,\ldots,n-1$)
.
For simplicity, we number the candidates according to their value, i.e., the $j$-th best candidate is simply candidate $j$, for $j = 1,\ldots,r$. Note that none of the best $r$ candidates will ever be rejected, since they are part of the optimum that is accepted in the end. Hence, they will just be postponed or accepted. Let $R_j^i$ be the number of postponements of candidate $j$ in round $i$. By linearity of expectation, $\Ex{R_j} = \Ex{\sum_{i = 1}^{n-1} R_j^i} = \sum_{i=1}^{n-1} \Ex{R_j^i}$.
%

We start by analyzing each of the rounds $i=1,\ldots,r-1$ in Phase 1. If $j$ has not arrived until round $i$, then $\Ex{R_j^i} = 0$. If $j$ arrives before round $i$, then the postponements $\Ex{R_j^i}$ are a fair share of the number of non-unique arrivals of round $i$. Hence, if $j$ arrives before round $i$, then 
$$\Ex{R_j^i \mid j \text{ arrives newly before round } i} = \frac{1}{i} \cdot\left(\frac{n}{n-i} - 1\right) = \frac{1}{n-i}\enspace.$$
If $j$ is the $i$-th unique arrival, he will be postponed once more at the time of his unique arrival:
$$\Ex{R_j^i \mid j \text{ arrives newly in round } i} = 1 + \frac{1}{i} \cdot\left(\frac{n}{n-i} - 1\right) = 1 + \frac{1}{n-i}\enspace.$$
Clearly, $j$ is the $i$-th unique arrival with probability $1/n$, for every $i=1,\ldots,r-1$, so
\begin{align*}
\Ex{R_j^i} &= \frac{1}{n}\left(1 + \frac{1}{n-i}\right) + \frac{i-1}{n} \cdot \frac{1}{n-i} \; = \; \frac{1}{n} + \frac{i}{n(n-i)} \; = \; \frac{1}{n-i}\enspace.
\end{align*}

Let us now analyze the rounds $i=r,\ldots,n-r$ in Phase 2. Candidate $j$ will not be accepted during any of these rounds. If $j$ has not arrived until round $i$, then $\Ex{R_j^i} = 0$. If $j$ arrives before round $i$, then $R_j^i$ is a fair share of the non-unique arrivals of round $i$. Thus, if $j$ arrives before round $i$, then 
$$\Ex{R_j^i \mid j \text{ arrives newly before round } i} = \frac{1}{r} \cdot\left(\frac{n-i+r}{n-i} - 1\right) = \frac{1}{n-i}\enspace.$$
If $j$ is the $i$-th unique arrival, he will be postponed once more at the time of his unique arrival:
$$\Ex{R_j^i \mid j \text{ arrives newly in round } i} = 1 + \frac{1}{r} \cdot\left(\frac{n-i+r}{n-i} - 1\right) = 1 + \frac{1}{n-i}\enspace,$$
so, as above,
$$\Ex{R_j^i} = \frac{1}{n-i}\enspace.$$

In Phase 3, matters get slightly more complicated. Candidate $j$ gets accepted at the beginning of round $i$ if there are at most $k = i - (n-r) - 1$ strictly better candidates than $j$ that have arrived so far. Hence, candidate $j$ gets accepted at the latest in the beginning of round $n - r + j$, and therefore $\Ex{R_j^i} = 0$ for every $i = n-r-j,\ldots,n-1$. However, $j$ is quite likely to be accepted earlier, especially if $j$ is bounded away from 1 and $r$. 

In particular, the probability that $j$ has arrived before round $i$ is again $(i-1)/n$. Conditioned on the fact that $j$ has arrived in this way, let $X_j$ be the number of candidates from the set of the best $j-1$ candidates that are among the first $i$ unique arrivals. $X_j$ is distributed according to a hypergeometric distribution -- we draw $j-1$ times (positions of $j-1$ top candidates) out of an urn of $n-1$ balls (remaining unique arrivals except the one chosen for $i$), where we have $i-1$ blue balls (unique arrivals up to and including round $i$, excluding the one chosen for $j$) and $(n-1)-(i-1) = n-i$ remaining red balls. $X_j$ is the number of blue balls we draw.

If $X_j \le k = i-(n-r)-1$, then $j$ is accepted before round $i$, i.e. $R_j^i = 0$. Otherwise, $j$ is undecided in round $j$ and gets postponed further in round $i$. In round $i$ there are $n-i$ undecided candidates and $n-i$ unseen ones. Thus, the number of postponements until the $(i+1)$-th unique candidate arrives is distributed according to a negative binomial distribution with probability $1/2$. From these postponements, candidate $j$ is drawn a fair share of times. Thus, if $j$ arrives before round $i$, then
\begin{align*}
\Ex{R_j^i \mid j \text{ arrives newly before round } i} = \Pr(X_j > \ell) \cdot \frac{1}{n-i} \cdot 1 = \Pr(X_j \ge k+1) \cdot \frac{1}{n-i} \enspace.
\end{align*}
If $j$ arrives newly in round $i$, then similar to the arguments above
\begin{align*}
\Ex{R_j^i \mid j \text{ arrives newly in round } i} =  \Pr(X_j \ge k+1) \cdot \left(1 + \frac{1}{n-i}\right) \enspace,
\end{align*}
which implies
\[
\Ex{R_j^i} = \Pr(X_j \ge k+1) \cdot \frac{1}{n-i}\enspace.
\]
Now, if $k + 1 = i-(n-r) \le \Ex{X_j}$, it is quite likely that $X_j \ge k+1$ is true. Intuitively, this is true for the rounds in which $i$ is smallest, i.e., for rounds $i$ with
\[ i-(n-r) \le \Ex{X_j} = \frac{i-1}{n-1}\cdot(j-1)\]
or, equivalently, $i \le n-r + \ell^*$, where
\[ \ell^* = \frac{(n-r-1)}{n-j} \cdot (j-1)\enspace.\]
For the rounds $i = n-r+1, \ldots, n-r+ \lfloor \ell^* \rfloor$ we will upper bound the probability by 1 and hence
\[
\Ex{R_j^i} = \Pr(X_j \ge k+1) \cdot \frac{1}{n-i} \le \frac{1}{n-i}\enspace.
\]
Phases 1, 2 and the first part of Phase 3 can be combined. This yields
\[
\Ex{R_j} = \sum_{i=1}^{n-1} \Ex{R_j^i} = \sum_{i=1}^{n-r+j-1} \Ex{R_j^i}
       = \sum_{i=1}^{n-r+\lfloor \ell^*\rfloor} \Ex{R_j^i} + \sum_{i=n-r+\lfloor \ell^*\rfloor + 1}^{n-r+j-1} \Ex{R_j^i}\enspace.
\]
For the first term, we bound as follows:
\begin{align*}
\sum_{i=1}^{n-r+\lfloor \ell^*\rfloor} \Ex{R_j^i}
&\le 
\sum_{i=1}^{n-r+\lfloor \ell^*\rfloor} \frac{1}{n-i} \le \frac{1}{r- \lfloor \ell^*\rfloor} + \int_{x=1}^{n-r+\ell^*} \frac{1}{n-x} \, dx \\
&\le \frac{1}{r - \lfloor \ell^*\rfloor} - \ln(n - (n-r+\ell^*)) + \ln(n-1)\\
&= \frac{1}{r - \lfloor \ell^*\rfloor} + \ln\left(\frac{n-1}{r - \ell^*}\right)\\
&= \frac{1}{r - \lfloor \ell^*\rfloor} + \ln\left(\frac{n-1}{r - \frac{(n-r+1)(j-1)}{n-j}}\right)\\
&\le\frac{1}{r - \lfloor \ell^*\rfloor} + \ln\left(\frac{(n-1)(n-j)}{r(n-j) - (n-r+1)(j-1)}\right)\\
&\le \frac{1}{r - \lfloor \ell^*\rfloor} + \ln\left(\frac{n-j}{r - j + 1}\right) \\
&\le \frac{1}{r - j + 1} + \ln\left(\frac{n-j}{r - j + 1}\right)\enspace.
\end{align*}

For the last term, in rounds $n-r + \lfloor \ell^* \rfloor + 1, \ldots, n-r+j-1$ we note that the improvement is small since $j \le r \le n/2$. By simply upper bounding $\Pr(X_j \ge k+1) \le 1$ again, we obtain an overall bound of
\begin{align*}
\Ex{R_j} = \sum_{i=1}^{n-r+j-1} \Ex{R_j^i}
&\le 
\sum_{i=1}^{n-r+j-1} \frac{1}{n-i} \le \frac{1}{r-j+1} + \int_{x=1}^{n-r+j-1} \frac{1}{n-x} \, dx \\
&\le \frac{1}{r - j + 1} + \ln\left(\frac{n-1}{r - j + 1}\right)\enspace.
\end{align*}
Note that the bound $\ln\left(\frac{n-j}{r - j + 1}\right)$ is arguably more accurate, since the rounds $i > n+r+\lfloor \ell^* \rfloor$ have significantly decreased probability for $j$ to get postponed until round $i$. The rather direct bound $\ln\left(\frac{n-1}{r - j + 1}\right)$ differs only by an additive term of at most $\ln(2) < 1$.

\parabold{Bottom-$(n-r)$ candidates:} The $j$-th best candidates with $j=r+1,\ldots,n$ are never accepted, only postponed and rejected. We again consider the algorithm in phases. The first phase is composed of rounds $i = 1,\ldots,r$, in which no candidate gets rejected. By repeating the analysis above, we see that for these rounds
\[ \Ex{R_j^i} \le \frac{1}{n-i} \enspace. \]

The second phase now consists of the remaining rounds $r+1,\ldots,n$. If it has arrived, then candidate $n$ will definitely get rejected in the beginning of round $r+1$, candidate $n-1$ in round $r+2$, and candidate $j$ in round $(n-j)+(r+1)$. However, candidate $j$ is much more likely to get rejected earlier. Here our analysis must take this fact into consideration and extend the arguments made for Phase 3 above. 

For candidate $j$, we condition on the fact that $i$ has arrived in the first $i$ rounds. Then it gets rejected before or in the beginning of round $i$ if at least $r$ strictly better candidates have arrived. Let $Y_j$ be the number of candidates from the set of the best $j-1$ candidates that are among the first $i$ unique arrivals. $Y_j$ is again distributed according to a hypergeometric distribution -- we draw $j-1$ times out of an urn of $n-1$ balls, where we have $i-1$ blue balls and $n-i$ red balls. $Y_j$ is the number of blue balls we draw.

If $Y_j \ge r$, then $j$ does not get postponed in round $i$, i.e., $R_j^i = 0$. Otherwise, $j$ is undecided in round $i$ and gets postponed. Following the analysis for Phase 2 above, we see that for rounds $i = r+1,\ldots, n-r$
\[
 \Ex{R_j^i} = \Pr(Y_j \le r-1) \cdot \frac{1}{n-i}\enspace.
\]
Moreover, following the analysis for Phase 3 above, the same holds for rounds $i = n-r+1,\ldots,n$.

Now, if $r-1 \ge \Ex{Y_j}$, it is quite likely that $Y_j \le r-1$ is true. Intuitively, this is true for the rounds in which $i$ is smallest, i.e., for rounds with
\[
r - 1 \ge \frac{i-1}{n-1} \cdot(j-1)
\]
or, equivalently, $i \le r + \ell^*$, where
\[
 \ell^* = (n-j) \frac{r-1}{j-1} \enspace.
\]
Note that this implies $r+\ell^* < n-j+r$. For the rounds $i = r+1,\ldots,r+\lfloor \ell^*\rfloor$, we upper bound the probability by 1 and hence
\[
\Ex{R_j^i} = \Pr(Y_j \le r-1) \cdot \frac{1}{n-i} \le \frac{1}{n-i}\enspace.
\]
Thus, combining Phase 1 and the first part of Phase 2, we get
\begin{align*}
\sum_{i=1}^{r+\lfloor \ell^* \rfloor} \Ex{R_j^i} &\le \sum_{i=1}^{r+\lfloor \ell^* \rfloor}\frac{1}{n-i} \le \frac{1}{n-r-\lfloor \ell^*\rfloor} + \int_{x=1}^{r+\ell^*} \frac{1}{n-x} dx\\
&\le \frac{1}{n-r-\lfloor \ell^* \rfloor} - \ln(n-(r+\ell^*)) + \ln(n-1)\\
&\le \frac{1}{n-r-\lfloor \ell^* \rfloor} + \ln\left(\frac{n-1}{n-r-\ell^*}\right)\\
&= \frac{1}{n-r-\lfloor \ell^* \rfloor} + \ln\left(\frac{n-1}{n-r-\frac{r-1}{j-1}(n-j)}\right)\\
&= \frac{1}{n-r-\lfloor \ell^* \rfloor} + \ln\left(\frac{(n-1)(j-1)}{(n-r)(j-1)-(r-1)(n-j)}\right)\\
&= \frac{1}{n-r-\lfloor \ell^* \rfloor} + \ln\left(\frac{j-1}{j-r}\right) \\
&\le \quad  \frac{1}{j-r} + \ln\left(\frac{j-1}{j-r}\right)\enspace.\\
\end{align*}
For the rounds $i = r+ \lfloor \ell^* \rfloor + 1, \ldots, n-j+r$, we first consider $r=1$. Here we have $\Pr(Y_j \le 0) = {n-i \choose j-1} / {n-1 \choose j-1} \le \left(\frac{n-i}{n-1}\right)^{j-1}$, where the latter is the probability of getting no blue balls when we draw with replacement. This probability is obviously higher, since we replace the red balls upon drawing them. Now we know $j \ge r+1 = 2$, which gives
\[
\sum_{i=r+\lfloor \ell^* \rfloor + 1}^{n-j+r} \Ex{R_j^i} \le 
\sum_{i=r+\lfloor \ell^* \rfloor + 1}^{n-j+r} \left(\frac{n-i}{n-1}\right)^{j-1} \cdot \frac{1}{n-i} \le
\sum_{i=r+\lfloor \ell^* \rfloor + 1}^{n-j+r} \frac{1}{n-1} < 1\enspace.
\]
This proves the theorem for $r=1$.

When $r > 2$ we bound the sum using a tail bound%
\footnote{Interestingly, when applying the same analysis using the more prominent but more coarse tail bound $\Pr(Y_j \le \Ex{Y_j} - t(j-1)) \le \exp(-2t^2 (j-1))$ it seems impossible to obtain a constant bound.} %
for the hypergeometric distribution
\[
 \Pr(Y_j \le \Ex{Y_j} - t(j-1)) \le \exp(-(j-1) \cdot {\sf KL}(p-t||p))
\]
where 
\[
{\sf KL}(a,b) = a \ln\left(\frac{a}{b}\right) + (1-a)\ln\left(\frac{1-a}{1-b}\right)\enspace.
\]
Here we use $p = \frac{i-1}{n-1}$, $t = \frac{i-1}{n-1} - \frac{r-1}{j-1}$, which implies for round $i$ that $r-1 = \Ex{Y_j} - t(j-1)$. Note that $t < p$ since $r \ge 2$. Now we define $\alpha = (r-1)/(j-r)$. With $j \ge r+1$ and $r\ge 2$ we have $\frac{1}{n-2} \le \alpha \le n-2$ and 
\begin{align*}
&\Pr(Y_j \le r-1) \le \exp(-(j-1) \cdot {\sf KL}(p-t||p)) \\
&= \exp\left(-(j-1) \cdot \frac{r-1}{j-1} \ln \frac{(r-1)(n-1)}{(j-1)(i-1)} - (j-1)\cdot\frac{j-r}{j-1} \cdot\ln\frac{(j-r)(n-1)}{(j-1)(n-i)}\right)  \\
&= \frac{((1+\alpha)(j-r))^{(1+\alpha)(j-r)}}{(\alpha(j-r))^{\alpha(j-r)} \cdot (j-r)^{j-r}} \cdot \left(\frac{i-1}{n-1}\right)^{\alpha(j-r)} \cdot \left( \frac{n-i}{n-1}\right)^{j-r} \\
&= \left(\frac{(1+\alpha)^{(1+\alpha)}}{\alpha^{\alpha}}\cdot \left(\frac{i-1}{n-1}\right)^{\alpha} \cdot \frac{n-i}{n-1}\right)^{j-r} \\
&\le \left(e \cdot (1+\alpha) \cdot \left(\frac{i-1}{n-1}\right)^{\alpha} \cdot \frac{n-i}{n-1}\right)^{j-r}\enspace.
\end{align*}
Note that
\[ r+\lfloor \ell^* \rfloor + 1 \ge r + \lceil \ell^* \rceil \ge \left\lceil n\cdot\frac{r-1}{j-1} \right\rceil = \left\lceil n\cdot\frac{\alpha}{1+\alpha} \right\rceil\enspace.\]
Hence,
\begin{equation}
\sum_{i=r+\lfloor \ell^* \rfloor + 1}^{n-j+r} \Ex{R_j^i} \le
\sum_{i=\lceil n\alpha/(1+\alpha)\rceil }^{n-j+r} \left(e \cdot (1+\alpha) \cdot \left(\frac{i-1}{n-1}\right)^{\alpha} \cdot \frac{n-i}{n-1}\right)^{j-r} \cdot \frac{1}{n-i}\enspace.
\label{eq:start}
\end{equation}
%
%
%

To provide a constant upper bound on \eqref{eq:start}, we first consider the case when $j = r + 1$. Observe that in this case $j-r = 1$ and $\alpha = r-1$. The formula simplifies to
\begin{align*}
&\sum_{i=r+\lfloor \ell^* \rfloor + 1}^{n-j+r} \Ex{R_j^i}\\
&\le  e \cdot r \cdot \frac{1}{n-1} \cdot \sum_{i=\lceil n(1-1/r) \rceil}^{n-1} \left(\frac{i-1}{n-1}\right)^{r-1} \\
&\le e \cdot \frac{r}{n-1} \cdot \left( \left(\frac{n-2}{n-1}\right)^{r-1} +  \int_{x=n(1-1/r)}^{n-1} \frac{n-1}{r} \left(\frac{x-1}{n-1}\right)^{r} dx \right)\\
&\le e \cdot \frac{r}{n-1} \cdot \left(1 + \frac{n-1}{r} \left(\left(\frac{n(1-2/n)}{n-1}\right)^{r} - \left(\frac{n(1-1/r)}{n-1}\right)^r\right) \right)\\
&< e \cdot \frac{r}{n-1} + e \cdot \left(1+\frac{1}{n-1}\right)^r\left(1 - \frac{1}{e}\right) \\
&\le e\cdot\left(\frac{r}{n-1} + e-1 \right) \le e^2\enspace.
\end{align*}
This proves the theorem for $r \ge 2$ and $j = r+1$.

%
%

Finally, for $r\ge 2$ and $j \ge r+2$, we split up the sum using values $n\alpha/(1+\alpha) = a_0 \le a_1 \le \ldots \le a_k \le a_{k+1} \le \ldots$ with
\[ a_k = n \cdot \left(1 - \frac{1}{e^k(1+\alpha)}\right)\enspace.\]
Note that $a_k$ are chosen such that over the interval $[a_k,a_{k+1}]$ the function $\frac{1}{n-i}$ at most differs by a factor of $e$, i.e., $\frac{e}{n-a_k} = \frac{1}{n-a_{k+1}}$.

For the rounds $i \in [a_0, a_1]$ we again use an upper bound $\Pr(Y_j \le r-1) \le 1$ and obtain
\begin{align*}
&\sum_{i=r+\lfloor \ell^* \rfloor + 1}^{n-j+r} \Ex{R_j^i}\\
&\le \sum_{i = \lceil a_0 \rceil }^{\lfloor a_1 \rfloor} \frac{1}{n-i} + 
\sum_{k=1}^\infty \sum_{i = \lceil a_k \rceil }^{\lfloor a_{k+1} \rfloor}  \left(e\cdot (1+\alpha) \cdot \left(\frac{x-1}{n-1}\right)^{\alpha} \cdot \frac{n-x}{n-1}\right)^{j-r} \cdot \frac{1}{n-i} \\
&\le \frac{1}{j-r} + 1 + \sum_{k=1}^\infty \sum_{i = \lceil a_k \rceil }^{\lfloor a_{k+1} \rfloor}  \left( e \cdot (1+\alpha) \cdot \left(\frac{x-1}{n-1}\right)^{\alpha} \cdot \frac{n-x}{n-1}\right)^{j-r} \cdot \frac{1}{n-i} \enspace.
\end{align*}
For the remaining sum, we bound each term for $k = 1,2,\ldots$ separately by
\begin{align*}
&\sum_{i=\lceil a_k \rceil }^{\lfloor a_{k+1} \rfloor }  \left(e \cdot(1+\alpha) \cdot \left(\frac{i-1}{n-1}\right)^{\alpha} \cdot \frac{n-i}{n-1}\right)^{j-r} \cdot \frac{1}{n-i} \\
&\le \sum_{i=\lceil a_k \rceil }^{\lfloor a_{k+1} \rfloor }  \frac{1}{n-i} \cdot
\left(e \cdot (1+\alpha) \cdot \left(\frac{a_{k+1}-1}{n-1}\right)^{\alpha} \cdot \frac{n-a_k}{n-1}\right)^{j-r}\\
%
&\le \sum_{i=\lceil a_k \rceil }^{\lfloor a_{k+1} \rfloor }  \frac{1}{n-i} \cdot
\left((1+\alpha) \cdot e \cdot \frac{n-a_k}{n-1}\right)^{j-r}\\
&= \sum_{i=\lceil a_k \rceil }^{\lfloor a_{k+1} \rfloor }  \frac{1}{n-i} \cdot
\left(e \cdot \left(\frac{n}{(n-1)e^k}\right)\right)^{j-r}\\
&\le \frac{e}{e^{(k-1)(j-r)}} \cdot \sum_{i=\lceil a_k \rceil }^{\lfloor a_{k+1} \rfloor }  \frac{1}{n-i}\enspace.
\end{align*}
Now using a standard upper bound via the integral,
\[
\sum_{i=\lceil a_k \rceil }^{\lfloor a_{k+1} \rfloor }  \frac{1}{n-i}
\quad \le \quad \frac{1}{n-\lfloor a_{k+1} \rfloor} + \int_{x= a_k }^{a_{k+1}} \frac{1}{n-x} \; dx \quad \le \quad \frac{e^{k+1}(1+\alpha)}{n} + 1\enspace. \]
Since $j-r \ge 2$ we obtain
\begin{align*}
& \sum_{k=1}^\infty \sum_{i=\lceil a_k \rceil}^{\lfloor a_{k+1} \rfloor}  \left(e \cdot (1+\alpha) \cdot \left(\frac{i-1}{n-1}\right)^{\alpha} \cdot \frac{n-i}{n-1}\right)^{j-r} \cdot \frac{1}{n-i} \\
&\le \sum_{k=1}^\infty \frac{e}{e^{(k-1)(j-r)}} \cdot \left( \frac{e^{k+1}(1+\alpha)}{n} + 1 \right) \\
&< \sum_{k=1}^\infty \frac{e^3/2}{e^{(k-1)(j-r-1)}} + \frac{e}{e^{(k-1)(j-r)}} \\
&= \sum_{k=0}^\infty \frac{e^3/2}{e^{k(j-r-1)}} + \frac{e}{e^{k(j-r)}} \\ 
&= \frac{e^3}{2} \left(1+\frac{1}{e^{j-r-1} - 1}\right) + e \left(1+\frac{1}{e^{j-r} - 1}\right)\enspace.  
\end{align*}
Hence, overall for the case $j-r \ge 2$
\[
  \sum_{i=r+\lfloor \ell^* \rfloor + 1}^{n-j+r} \Ex{R_j^i} <  \frac{e^3}{2} \left(1+\frac{1}{e^{j-r-1} - 1}\right) + e \left(1+\frac{1}{e^{j-r} - 1}\right) + 1 + \frac{1}{j-r} \enspace.
\]
The term is constant and the theorem follows. For the sake of clarity in the analysis we did not attempt to optimize any constants. Based on our experiments (see Figure~\ref{fig:postponement}), it seems possible to reduce the term $O(1)$ significantly by increasing the technical overhead in the analysis. 
\qed

\subsection{Proof of Theorem~\ref{thm:logToNloglog}}
For any $\varepsilon > 0$, consider a complete $m$-partite graph with $m = n/(3 \ln n)$. The nodes arrive sequentially in random order, and every node has an arbitrary positive but otherwise unknown value. Clearly, the optimum independent set $I^*$ consists of the nodes of exactly one partition, and $|I^*| = n/m = 3 \ln n$.

The adversary assigns the values of nodes in a small but unknown interval $[1,1+\delta]$. Thus, unless every node of a partition has arrived, the nodes in that partition cannot be rejected by \MaintainOPT, since the last node could result in that partition becoming the optimal one. The nodes of the optimal partition can only be accepted in the very end when all nodes have been seen. A partition of nodes can be rejected only if it has fully arrived, and only if there is at least one other partition for which the arrived nodes have larger total value.

We again analyze \MaintainOPT\ in rounds based on the unique arrivals as in the previous proofs. Consider the status at the end of round $k = n-1-m$. There are $m$ unseen candidates. Our idea is to determine these candidates by throwing $m$ balls (unseen candidates) into $n^{1-\varepsilon}$ bins (partitions). Then in each partition we pick the subset of candidates from that partition randomly. Note that the equivalence between our arrival process and the balls-and-bins process breaks only if we throw too many balls into a bin. After all, each partition (bin) contains only $\ln n$ many candidates (space for balls). However, if we throw $m = n/ (3 \ln n)$ balls into $m$ bins, standard calculations show that (1) with probability at most $\frac{2}{m}$, there is at least one bin with load at least $(3 \ln m)/(\ln \ln m)$; and (2) with probability at most $\frac{2}{m}$ the number of empty bins $\ell_e$ is  $|\ell_e - \frac{m}{e}| > \sqrt{m} \ln m$. Thus, by a simple union bound, we have that with probability $1-\frac{4}{m}$, the balls-and-bins process returns an assignment where every bin contains at most $3 \ln m = 3 \ln (n/ (3 \ln n)) < 3 \ln n$ balls and the number of empty bins is at most $\frac{m}{e} + \sqrt{m} \ln m$. 

Therefore, in our arrival process at the end of round $k$ with probability $1-\frac{4}{m}$ there are at least $d = m - \frac{m}{e} + \sqrt{m} \ln m$ partitions, for which there is at least one unseen candidate. Since all candidates in these partitions must remain undecided, there are at least $\ell = n/m \cdot d = \left(1 - \frac{1}{e}\right) \cdot n - o(n)$ undecided candidates at the end of round $k$.

Hence, with probability $1-\frac{12 \ln n}{n}$, there are $\ell$ undecided candidates throughout the first $k$ rounds, in which case the expected number of postponements in rounds $i = 1,\ldots,k$ is at least
\begin{align*}
\sum_{i=1}^k \Ex{R^i} &\ge \sum_{i=1}^{k} \left(\frac{\ell}{n-i} - 1\right) \ge \ell \cdot \ln\left(\frac{k}{n-k}\right)
\ge \ell \cdot \ln(3 \ln n - 2) = \Omega(n \ln \ln n)\enspace.
\end{align*}
\vspace{-1.25cm}
\qed

\end{document}